\pdfoutput=1
\documentclass[compress]{elsarticle}
\makeatletter
\def\ps@pprintTitle{%
 \let\@oddhead\@empty
 \let\@evenhead\@empty
 \def\@oddfoot{\centerline{\thepage}}%
 \let\@evenfoot\@oddfoot}
\makeatother
\usepackage[utf8]{inputenc}

\usepackage[normalem]{ulem}

\usepackage{breakurl}
\usepackage{booktabs}
\usepackage{url}
\usepackage{natbib}
\usepackage{graphicx}
\usepackage{amsmath}
\usepackage{amsthm}
\usepackage{microtype}
\usepackage{flushend}

\DeclareMathOperator{\division}{division}
\DeclareMathOperator{\remainder}{remainder}
\DeclareMathOperator{\floor}{floor}
\DeclareMathOperator{\ceiling}{ceiling}

\newtheorem{remark}{Remark}
\newtheorem{lemma}{Lemma}
\newtheorem{theorem}{Theorem}

\newtheorem{proposition}{Proposition}

\usepackage{todonotes}

\newcommand{\ourtitle}{Integer Division by Constants: Optimal Bounds}
\long\def\ignore#1{}

\journal{Information Processing Letters}
 \author[UQAM]{Daniel Lemire\corref{cor1}} \ead{daniel.lemire@teluq.ca}
 \author[UQAM]{Colin Bartlett} \ead{natekurz@gmail.com}
 \author[Upscaledb]{Owen Kaser} \ead{chris@crupp.de}

 \address[UQAM]{\scriptsize  Universit\'e du Qu\'ebec (TELUQ),  5800 Saint-Denis, Montreal, Quebec, Canada
}
\address[Upscaledb]{\scriptsize Computer Science Department, UNB Saint John, New Brunswick, Canada}

 \cortext[cor1]{Corresponding author. Tel.: 00+1+514 843-2015
   ext. 2835; fax: 00+1+800 665-4333.}

\begin{document}
\title{\ourtitle{}}

\begin{frontmatter}

\begin{abstract}









The integer division of a numerator $n$ by a divisor $d$ gives a  quotient $q$ and a remainder $r$.  Optimizing compilers accelerate software by replacing the division
of $n$ by $d$ with the division of  $c * n$ (or $c * n + c$) by $m$ for convenient integers $c$ and $m$ chosen so that they approximate the reciprocal: $c/m\approx 1/d$.  
Such techniques are especially advantageous when $m$ is chosen to be a power of two and when $d$ is a constant so that $c$ and $m$ can be precomputed.  
The literature contains many bounds on the distance between $c/m$ and the divisor $d$. Some of these bounds are optimally tight, while others are not. We present optimally tight bounds for quotient and remainder computations.
\end{abstract}


\end{frontmatter}


\section{Introduction}

The problem of computing the integer division given constant divisors has a long history in computer science~\cite{Artzy:1976:FDT:359997.360013,jacobsohn1973combinatoric,Li:1985:FCD:4135.4990,Vowels:1992:D}.
Granlund and Montgomery~\cite{Granlund:1994:DII:773473.178249} present the first general-purpose algorithms to divide integers by constants using a multiplication and a division by a power of two: their work was adopted by the GNU Compiler Collection (GCC). Given any non-zero 32-bit divisor known at compile time, the optimizing compiler can replace the division by a multiplication followed by a shift. 
Warren~\cite{warr:hackers-delight-1st} improved on the Granlund and Montgomery technique by deriving a better bound that gives a wider range of choices.
 Warren's  better approach is found in LLVM's Clang compiler.
 Many optimizing compilers rely on equivalent techniques, either based on the original Granlund-Montgomery article or on Warren's technique. 
 
 Robison~\cite{Robison:2005:NUD:1078021.1078059} describes a slightly superior alternative for some divisors in that we multiply and add the multiplier before dividing by a power of two (henceforth the multiply-add technique). Though it comes at the cost of an  addition, it allows one to choose a smaller multiplier, which can be advantageous.  Robison's approach is implemented in the popular libdivide library~\cite{fish2011}.
 
Most of the literature is focused on the computation of the quotient $q$ of the division of $n$ by $d$. From the quotient $q$, we can compute the remainder as $n-q * d$. We can also compute the remainder directly~\cite{lemire2019faster} without first computing the quotient: it is given by taking remainder of $c*n$ divided by $m$, and then multiplying it by $m$. However, for the remainder and the quotient to be exact, it is necessary that $c/m$ approximates $1/d$ more closely than if we merely need the quotient.

From the computation of remainders, we can derive a divisibility check, that is check whether $d$ divides $n$, or, equivalently, check that $n$ is a multiple of $d$. Though it may seem that computing the remainder and checking whether it is zero is efficient, we can  simplify and accelerate the algorithm by avoiding the computation of the remainder.

The literature commonly assumes that $m$ is a power of two. We approach the problem more generally, letting $m$ and $c$ be any integer, and restricting the numerator to an interval $[0,N]$ where $N$ can be any integer. It makes our exposition more general, while simplifying the notation.

Some our novel contributions are as follows:
\begin{itemize}
\item We improve Robison's bound~\cite{Robison:2005:NUD:1078021.1078059}, in a manner similar to how Warren improved Granlund and Montgomery's bound. That is, we provide an optimal  bound for the multiply-add technique.\footnote{Drane et al.~\cite{drane2012correctly} have a related bound, but they also have additional constraints on the divisor $d$.}
\item We derive a new tighter bounds for computing the quotient directly and checking the divisibility, thus improving on the work of Lemire at al.~\cite{lemire2019faster}
\item We show that we can adapt Robison's technique to compute remainders directly and derive a novel bound. We  adapt the multiply-add technique for the purpose of a divisibility check. To our knowledge, these results are  novel.
\end{itemize}
All our bounds on how close $c/m$ must be to $1/d$ are optimal and form necessary and sufficient conditions. Table~\ref{table:bigsummary} presents  our core results in concise manner.

\begin{table*}[tbh]\centering
\caption{\label{table:bigsummary} Summary of main results. Throughout, all values are non-negative integers, the divisor is non-zero $d>0$, and the numerator is bounded by $N\geq d$, so that $n\in [0,N]$. We add the constraint that $c \in[0,m)$ so that $c$ is as small as possible.}
\small

\begin{tabular}{|p{1cm}p{1.5cm}p{8.5cm}|}
\hline
\multicolumn{3}{|l|}{Theorem~\ref{thm:quotient} and  Warren~\cite{warr:hackers-delight-2e}}\\
 & statement: &
$\division(n,d) =\division(c*n , m)$ for all $n\in [0,N]$ \\
& condition: & $1/d \leq c/m < \left (1 + \frac{1}{N- \remainder(N+1,d)} \right )1/d$ \\\hline
\multicolumn{3}{|l|}{Theorem~\ref{thm:remainder} (novel, improves Lemire et al.~\cite[Theorem~1]{lemire2019faster})}\\
&statement: &
$\division(n,d) =\division(c*n , m)$ and $\remainder(n,d) =\division(\remainder(c*n , m) * d , m)$ for all $n\in [0,N]$ \\
& condition: & $1/d \leq c/m< \left (1+\frac{1}{N} \right )1/d$ \\\hline
\multicolumn{3}{|l|}{Proposition~\ref{prop:divisibility} (novel, generalizes Lemire et al.~\cite{lemire2019faster})}\\
&statement: &
 $d$ divides $n\in[0, N]$  if and only if 
$\remainder(c*n , m)  < c$ \\
& condition: & $1/d \leq c/m< \left (1+\frac{1}{N} \right )1/d$ \\\hline
\multicolumn{3}{|l|}{Theorem~\ref{thm:quotientrobison}  (improves Robison~\cite{Robison:2005:NUD:1078021.1078059})}\\
&statement: &
$\division(n,d) =\division(c*n +c , m)$ for all $n\in [0,N]$ \\
& condition: & $\left (1-  \frac{1}{N- \remainder(N,d) + 1} \right ) 1/d \leq c/m < 1/d$ \\\hline
\multicolumn{3}{|l|}{Theorem~\ref{thm:remainderrobison} (novel)}\\
&statement: &
$\division(n,d) =\division(c*n +c , m)$ and $\remainder(n,d) =\division(\remainder(c*n +c, m) * d , m)$ for all $n\in [0,N]$ \\
& condition: & $\left (1-\frac{1}{N+1} \right )1/d \leq c/m< 1/d$ \\\hline
\multicolumn{3}{|l|}{Proposition~\ref{prop:multiplyadddivisibility} (novel)}\\
&statement: &
 $d$ divides $n\in[0, N]$  if and only if 
$\remainder(c*n+c, m)  < c$ \\
& condition: & $\left (1-\frac{1}{N+1} \right )1/d \leq c/m<  1/d$ \\\hline
\end{tabular}

\end{table*}

\section{Other Related Work}

The problem of quickly computing the division by a constant in computers dates back to at least the 1970s.
Jacobsohn~\cite{jacobsohn1973combinatoric} shows that we can divide by an odd integer  by multiplying by a fractional inverse, followed by some rounding.
Artzy et al.~\cite{Artzy:1976:FDT:359997.360013} describe a related algorithm to divide multiples of a known divisor (exact division). Li~\cite{Li:1985:FCD:4135.4990} presents algorithms for integer division by all odd integers up to 55~\cite[\S~10-18]{warr:hackers-delight-2e}.
Divisions are executed as series of ``shift and add'' instructions.

Magenheimer et al.~\cite{Magenheimer:1987:IMD:36177.36189}
describe how to compute the division of
integers by odd divisors  as a multiplication and an addition followed by a division by a power of two. Their approach was later refined by Robison~\cite{Robison:2005:NUD:1078021.1078059}. Similarly, 
Granlund and Montgomery's approach~\cite{Granlund:1994:DII:773473.178249} (without an intermediate addition) was refined by  Cavagnino and Werbrouck~\cite{Cavagnino:2008:EAI:1388169.1388172}, and later by Warren~\cite{warr:hackers-delight-2e}. As remarked by Robison~\cite{Robison:2005:NUD:1078021.1078059}, the two approaches (with and without an intermediate addition) are complementary: we can choose one or the other depending on the divisor. We review and elaborate on this complementarity in \S~\ref{sec:complementarity}.

To our knowledge, the latest work on the software acceleration of the division by constants was Lemire et al.~\cite{lemire2019faster}. They revisited two specific problems: the direct computation of the remainder---without first computing the quotient---and the related divisibility tests. Compared to optimizing compilers that compute the remainder by first computing the quotient, they found that their direct approach could be up to 30\% faster. Their divisibility test could be twice as fast as the code produced by popular optimizing compilers and libraries. It can also be up to twice as fast as the state-of-the-art  divisibility check proposed by Granlund and Montgomery~\cite{Granlund:1994:DII:773473.178249}.
They did not consider the multiply-add approach, a gap that we fill with \S~\ref{sec:multiply-add}. We also make their main result~\cite[Theorem~1]{lemire2019faster} tighter (see Theorem~\ref{thm:remainder}). We similarly improve mathematically on their divisibility check~\cite[Proposition~1]{lemire2019faster} (see Proposition~\ref{prop:divisibility}). Our improvements may not immediately result in improved software performance, but they fill a conceptual gap.  The systematic computation of the remainder directly as proposed by Lemire et al., without first computing the quotient, has received  attention in the hardware and circuit literature~\cite{331617,Doran1995,7933010,deDinechin2012} but had never been generally exploited in software as far as we know. One practical exception was the work by Vowels~\cite{Vowels:1992:D} who described the direct computation of both the quotient and remainder, in the special case where we divide by 10.

\section{Technical Preliminaries}

For non-negative real numbers $z$, $\floor(z)$ is the greatest integer no larger than $z$. It is a monotonic function: if $z_1 \geq z_2$ then $\floor(z_1)\geq \floor(z_2)$.

We define $\division(x , y)\equiv  \floor(x/y)$
and  \begin{eqnarray}\remainder(x , y)&\equiv& x -\division(x,y)*y\\
&= &x - \floor(x/y)*y\end{eqnarray} for positive real numbers $x, y$ with the constraint that $y\neq 0$. We have that $\remainder(x , y)\in [0,y)$. If $y$ is an integer and $x$ is not an integer, then $\remainder(x , y)\neq 0$. By definition, we always have that $x= \division(x,y)*y+\remainder(x , y)$.

\begin{lemma}\label{lemma:buzz2}Consider a positive integer $d> 0$, a non-negative integer $n$ and a non-negative real number $x$. We have that 
$\remainder(n,d) = \floor(\remainder(x,d))$
    and 
$\division(n,d) = \division(x,d)$
if and only if 
$n\leq x <n+1.$
\end{lemma}
\begin{proof}
($\Leftarrow$) We can verify that if $n\leq x <n+1$, the previous two conditions are satisfied.

($\Rightarrow$) Assume that  
$\remainder(n,d) = \floor(\remainder(x,d))$
and 
$\division(n,d) = \division(x,d).$
We have 
\begin{eqnarray}\remainder(n,d) &\leq& \remainder(x,d)\\
&<& \remainder(n,d)+1.\end{eqnarray}
 By expanding out 
$\remainder(n,d) = \floor(\remainder(x,d)),$
we have that 
\begin{eqnarray}n - \floor(n/d)*d &\leq& x - \floor(x/d)*d  \\
&<& n - \floor(n/d)*d  + 1.\end{eqnarray}
Expanding out $\division(n,d)=\division(x,d)$ and multiplying by $d$, we have 
$\floor(n/d)*d = \floor(x/d)*d.$
We establish the lemma by adding this last equation to the previous inequality.
\end{proof}

\section{Multiply-Divide Results}
\label{sec:Multiply-Divide}
Given a non-negative numerator $n$ and non-zero divisor $d$,
we want to show that by choosing integer constants $c$ and $m$ carefully, we
can compute $\division(n,d)$ and $\remainder(n,d)$ by starting from $c*n$ and dividing by $m$.

\subsection{Quotient}
\label{sec:quotient}

We want to find $c$ and $m$ such that $\division(n,d) = \division(c*n , m)$. Intuitively, this equation implies that $n/d \approx c * n /m $ and $n \approx c * n *d/m $. Let us formalize this intuition.

For any non-negative real number  $x$ and non-negative integer $Q$,
we have that $\floor(x/d)=Q$ is equivalent to $x \in [Qd,Qd+d)$.  Letting
$Q=\floor(n/d)$ and $x = c*n*d/m$, we get $d*\floor(n/d) \leq c*n*d/m <
d*\floor(n/d)+d$.  Since $d*\floor(n/d)=n-\remainder(n,d)$, 
$\floor(x/d) = \division(c*n, m)$ and $Q= \division(n,d)$, we have that
$\division(n,d) = \division(c*n , m)$ is equivalent to
$n-\remainder(n,d)\leq c*n*d/m <n - \remainder(n,d) + d.$

Consider a range of  integer numerators $n \in [0,N]$ for some maximal  integer numerator $N\geq d$.  We want this equation to hold for all $n$.
The equation is satisfied trivially when $n=0$ and $d>0$. Suppose that $n>0$ and rewrite the inequalities as 
$(n-\remainder(n,d))/n  \leq c*d/m < (n - \remainder(n,d) + d)/n.$

Given any $n\in[1,N]$ for some integer $N\geq d$, we have that the leftmost expression
$ (n-\remainder(n,d))/n= 1-\remainder(n,d)/n$
is largest and equal to 1 when 
$\remainder(n,d)=0.$
Meanwhile the rightmost expression $1 + (d- \remainder(n,d))/n$ is smallest when $n$ is as large as possible with $\remainder(n,d)=d-1$.
To prove this bound, partition the possible numerators into sets 
$N_k = \{ n \in [0,N]\,|\,\remainder(n,d)=k\}.$
Fixing $N$ and $d$, we seek the value $n \in [1,N]$ minimizing 
$f(n)=1 + (d- \remainder(n,d))/n.$
For $n \in N_k$ we have $f(n)= 1 + (d- k)/n$ which is minimized for the largest member of $N_k$. 
Let $v$ be the largest member of $N_{d-1}$; we have $f(v)= 1+1/v$. 
We see that the values of $n$ in $[N-d+1, N]$ are  the minimizing values in each $N_k$. 
Among these, we can show $v$ minimizes $f$.
\begin{itemize}
\item Consider any $n \in [N-d+1,N]$ with $n> v$.  Write it as $n = v+k$ with $k> 0$ and $k\leq d-1$
and so $\remainder(n,d)=k-1$ and $f(n) = 1 + (d-k+1)/(v+k)$. As $k$ increases, the numerator decreases and the denominator increases, we have that 
the minimum is reached when $k$ is largest ($d-1$), in which case 
$f(n) = 1 + 2/(v +d -1)\geq 1+2/(v+v)= 1 + 1/v = f(v)$, since $v \geq d-1$.
\item Consider any $n \in [N-d+1,N]$ with $n<v$.
Write $n$ as $v-d+k$, so $\remainder(n,d)$ is again $k-1$.  We have $f(n)=1+(d-k+1)/(v-d+k)$. Again,   as $k$ increases, the numerator decreases and the denominator increases, we have that the minimum is reached when $k$ is largest ($d-1$) in which case $f(n)=1+2/(v-1) < 1 + 1/v = f(v)$.
\end{itemize}
Thus $n=v$ minimizes $f(n)$. We have shown Lemma~\ref{lemma:ohoh} because $v = N- \remainder(N+1,d)$.


\begin{lemma}\label{lemma:ohoh}Given an integer $d>0$, the value of $1 + (d- \remainder(n,d))/n$ over $n=0,1,\ldots, N$ is minimized when $n$ is   $n=N- \remainder(N+1,d)$.
\end{lemma}

Hence we have that
$1 \leq c*d/m < 1 + \frac{1}{N- \remainder(N+1,d)}$
is equivalent to $\division(n,d) = \division(c*n , m)$ for all $n\in [0,N]$.

\begin{theorem}\label{thm:quotient}
Consider an integer divisor $d>0$ and a range of  integer numerators $n\in[0, N]$ where $N\geq d$ is an integer.
We have that $\division(n,d) =\division(c*n , m)$ for all integer numerators $n$ in the range  if and only if
\begin{eqnarray}1/d \leq c/m < \left (1 + \frac{1}{N- \remainder(N+1,d)} \right )1/d.\end{eqnarray} 
\end{theorem}

\begin{remark}
Granlund and Montgomery~\cite{Granlund:1994:DII:773473.178249} have an upper bound of $c/m \leq (1+1/(N+1))/d$ as a sufficient (but not necessary) condition. A bound equivalent to Theorem~\ref{thm:quotient} is derived by Warren~\cite{warr:hackers-delight-1st}. 
\end{remark}

Once we have a pair of inequalities as in Theorem~\ref{thm:quotient}, we can solve for $c$ and $m$. It is always possible to do so: we can verify that
$c=1$, $m=d$ is always a solution. However, we may have further constraints on $c$ and $m$: maybe we require $m$ to be a power of two. We can show that as long as we can choose $m$ arbitrarily large, there is always a solution. Letting $K=N- \remainder(N+1,d)$, we can rewrite the inequalities as
$m/d \leq c < \left (1 + \frac{1}{K} \right ) \frac{m}{d}.$
Thus if $c$ is to be as small as possible, we must have that $c = \ceiling (m/d)$. It remains to solve for $m$ such that 
$\ceiling \left ( \frac{m}{d} \right ) < \left (1 + \frac{1}{K} \right )\frac{m}{d}.$
Because $\ceiling(m/d) - m/d<1$, we have that the inequality is always satisfied when 
$m \geq K * d = (N- \remainder(N+1,d)) * d.$
This bound indicates that it is always possible to find a solution, by picking $m$ large enough.

\subsection{Remainder}
\label{sec:remainder}
From the quotient $\division(c*n , m)$, we get the quotient of the division of  $n$ by $d$; it is maybe intuitive that we can derive the remainder of the division of  $n$ by $d$ from  $\remainder(c*n , m)$.

Formally, we want to find integer constants $c>0$ and $m>0$ such that for any integer numerator $n\in[0, N]$ and integer divisor $d>0$, we have that $\remainder(n,d) =\division(\remainder(c*n , m) * d , m)$. 

 If we find $c$ and $m$ such that $\remainder(n,d) =\division(\remainder(c*n , m) * d , m)$ is satisfied, then replacing $c$ with $c+m$ or $c+2m$ would still work: in fact $\remainder(c*n , m) = \remainder(c*n + k*m*n, m) = \remainder((c + k*m)*n , m)$ for any integer $k$. Thus we 
require $c$ to be in $[0, m)$. 

With this constraint ($c\in [0, m)$),  we are able to show (see Lemma~\ref{lemma:three}) that the ability to compute remainders via
$\remainder(n,d) =\division(\remainder(c*n , m) * d , m)$ implies 
that the quotient of $n$ divided by $d$ is given by  $\division(n,d) =\division(c*n , m)$.
Intuitively, it is strictly more difficult to compute the remainder than to compute the quotient. Hence, if we just need the remainder, and not the quotient, we cannot relax our conditions when $c\in [0,m)$.

\begin{lemma}\label{lemma:three}
Consider an integer divisor $d>1$.
Suppose that  we have integer constants $c$ and $m$ such that $c\in [0, m)$
and $\remainder(n,d) =\division(\remainder(c*n , m) * d , m)$ for all
numerators $n\in [0, N]$ then we must have that $\division(n,d) =\division(c*n , m)$.
\end{lemma}
\begin{proof}

When $n = 0$, we have that $\division(n,d) =\division(c*n , m)$ holds trivially. 
Since $c\in [0,m)$ then $c*(n+1) - c*n < m$ so when  $\division(c*n , m)$ increases following an increment of $n$ by one, it must increase by at most one. We just have to show
that it happens exactly when $\remainder(n,d)=0$.

We have  $c* n =\remainder(c*n , m) + m *  \division(c*n , m)$. The left side of this equation increases by $c$ exactly when $n$ is incremented by one. When $\division(c*n , m)$ increases by one, then it contributes $m$ to the right side. Since $m >c$, we have that an increase of $\division(c*n , m)$ corresponds to a decrease of $\remainder(c*n , m)$.

However, we have that $\remainder(n,d) =\division(\remainder(c*n , m) * d , m)$. 
From this equation, we have that whenever $\remainder(n,d)$ increases when we increment $n$ by one, then
$\remainder(c*n , m)$ must also increase. 
We know that when $n$ is incremented, then either  $\remainder(n,d)$ increases by one, or goes back to zero.  It is not possible for $\remainder(n,d)$  to increase if $\remainder(c*n , m)$  decreases: it must therefore be that a decrease in $\remainder(c*n , m)$ corresponds to  $\remainder(n,d)= 0$. 
Thus we have that an increase of $\division(c*n , m)$ following an increment of $n$  corresponds $\remainder(n,d)= 0$.  It follows that $\division(n,d) =\division(c*n , m)$.
\end{proof}


We still must derive the conditions on $c$ and $m$.
We can expand the condition that $\remainder(n,d) =\division(\remainder(c*n , m) * d , m)$  as follows:
\begin{eqnarray}
\remainder(n,d) &&= \floor\left (\frac{\left (c*n - \floor\left (\frac{c*n}{m}\right) *m\right) * d}{ m}\right)  \\
&&= \floor\left (c*n*d/m - \floor\left (\frac{c*n*d/m}{d}\right) *d\right) \\
&&= \floor (\remainder(c*n*d/m,d)). 
\end{eqnarray}

Then by Lemma~\ref{lemma:buzz2}, we have that the two constraints ($\remainder(n,d) = \floor (\remainder(c*n*d/m,d))$
and $\division(n,d) =\division(c*n*d/m , d)$) are equivalent to $n\leq c*n*d/m <n+1$, or  $m/d \leq c  < (1+1/n)m/d$. This condition should hold for all applicable values of $n$, and thus we choose to use 
 the maximal value of $n$ (i.e., $N$) as it
 provides the tightest bound --- so an equivalent expression is  $m/d \leq c < (1+1/N)m/d$.

We have derived the following theorem.

\begin{theorem}\label{thm:remainder}
Consider an integer divisor $d>0$ and a range of  integer numerators $n\in[0, N]$ where $N\geq d$ is an integer.
We have that  $\division(n,d) =\division(c*n , m)$  and $\remainder(n,d) =\division(\remainder(c*n , m) * d , m)$ for all integer numerators $n$ in the range if and only if
\begin{eqnarray}1/d \leq c/m< \left (1+\frac{1}{N} \right )1/d.\end{eqnarray}
\end{theorem}

We can check that the conditions of Theorem~\ref{thm:remainder} are always met with  $c = \ceiling (m/d)$ and $m\geq N * d$. 

\begin{remark}
In previous work~\cite[Theorem~1]{lemire2019faster}, Lemire et al. reported an upper bound of  $c/m \leq (1+1/(N+1)) 1/d$ as a sufficient  (but not necessary) condition.
For the difference to matter, we need that there is an integer in 
the interval $((1+1/(N+1)) m/d,(1+1/N) m/d)$.
It happens in some instances, for example if
$N = 10$, $d = 5$, and $m=2^5$, we have that $7\in((1+1/(N+1)) m/d,(1+1/N) m/d)$. However, in the previous work~\cite{lemire2019faster}, Lemire et al. considered only the case where $m=N+1$. We can show that if $m=N+1$ and $m\geq 3$ then the earlier bound is tight. Indeed, if there is an integer $z$ in $((1+1/(N+1)) m/d,(1+1/N) m/d)$ then there must be an integer $z * d$ in $((1+1/(N+1)) m,(1+1/N) m)$. Substituting $N=m-1$, the interval becomes $(m+1, m+1+ 1/(m-1))$: because $m+1$ is an integer and $1/(m-1)\leq 1/2$, there is no integer in this interval. Hence, there cannot be an integer in $((1+1/(N+1)) m/d,(1+1/N) m/d)$ and the earlier bound is tight.
\end{remark}

If we only desire the remainder, and not the quotient, we can lift the restriction that $c \in [0, m)$: we can replace $c$ by $c+k*m$ for any integer $k$.

\subsection{Check for Divisibility}
\label{sec:divisibililty}
We have that $n$ is a multiple of $d$ if and only if $\remainder(n,d) = 0$.
Given Theorem~\ref{thm:remainder},  we can check whether $\remainder(n,d) = 0$ by checking whether 
$\division(\remainder(c*n , m) * d , m) = 0$.
In turn, we have that this last equation holds if and only if $\remainder(c*n , m) * d < m$ or  $\remainder(c*n , m)  < m/d$. 
Thus  $\remainder(c*n , m)  < m/d$ is a divisibility test. However, we  show the more elegant result that
$\remainder(c*n , m)  < c$ is a divisibility test (see Proposition~\ref{prop:divisibility}).

By the assumption of Theorem~\ref{thm:remainder}, we have that $m/d \leq c$. Thus if $n$ is a multiple of $d$, then we have that $\remainder(c*n , m)  < c$.
We need to prove the counterpart, that  $\remainder(c*n , m)  < c$ implies that
$n$ is a multiple of $d$.
By Theorem~\ref{thm:remainder}, we have that  $\division(n,d) =\division(c*n , m)$. Hence we have that $\division(c*n , m) =\division(c*(n - \remainder(n,d)) , m)$ since $n$ and $n - \remainder(n,d)$ have the same quotient with respect to $d$.
When two values $z_1, z_2$ have the same quotient ($\division(z_1, m)=\division(z_2, m)$)
then their difference must be captured by their remainders: 
$z_2-z_1 = \remainder(z_2, m)-  \remainder(z_1, m)$.
In this case, taking $z_1=c*n$ and $z_2=c*(n-\remainder(n,d))$,we have that their difference is $c* \remainder(n,d)$.
It follows that $\remainder(c*n , m) -\remainder(c*(n - \remainder(n,d)), m) = c*\remainder(n,d)$ and therefore $c*\remainder(n,d) \leq \remainder(c*n , m)$.
Thus if $\remainder(c*n , m)  < c$, we have  $c*\remainder(n,d) < c$ which implies  $\remainder(n,d) = 0$.

\begin{proposition}\label{prop:divisibility}
Consider an integer divisor $d>0$. We have that $d$ divides $n\in[0, N]$  if and only if 
$\remainder(c*n , m)  < c$
subject to the condition that
\begin{eqnarray}1/d \leq c/m< \left (1+\frac{1}{N} \right )1/d.\end{eqnarray}
\end{proposition}

Proposition~\ref{prop:divisibility} selects a value of $c$ in $[0,m)$ when $d>1$.

\section{Multiply-Add-Divide Results}
\label{sec:multiply-add}
Some authors~\cite{Robison:2005:NUD:1078021.1078059, Magenheimer:1987:IMD:36177.36189} have considered the case where we replace the division by a formula of the multiply-add form $\floor((c * n + b)/m)$ for some $b$. The benefit of the multiply-add approach is that it may allow one to pick a smaller value of $c$, compared to the simpler form
$\floor(c * n/m)$.
The derivations are nearly identical as in \S~\ref{sec:Multiply-Divide}, so we just give our results.

\begin{theorem}\label{thm:quotientrobison}
Consider an integer divisor $d>0$ and a range of  integer numerators $n\in[0, N]$ where $N\geq d$ is an integer.
We have that $\division(n,d) =\division(c*n +c , m)$ for all integer numerators $n$ in the range  if and only if
\begin{eqnarray} \left (1-  \frac{1}{N- \remainder(N,d) + 1} \right ) 1/d \leq c/m < 1/d.\end{eqnarray} 
\end{theorem}

\begin{remark}Robison~\cite{Robison:2005:NUD:1078021.1078059} derived the sufficient condition $(1-1/(N+1)) 1/d  \leq c/m < 1/d$. When $\remainder(N,d)\neq 0$, Robison's bound is suboptimal unlike Theorem~\ref{thm:quotientrobison}. Drane et al.~\cite{drane2012correctly} derive a similar result to ours for the case where $d$ is odd.
\end{remark}

\begin{theorem}\label{thm:remainderrobison}
Consider an integer divisor $d>0$ and a range of  integer numerators $n\in[0, N]$ where $N\geq d$ is an integer.
We have that  $\division(n,d) =\division(c*n+c , m)$  and $\remainder(n,d) =\division(\remainder(c*n +c , m) * d , m)$ for all integer numerators $n$ in the range if and only if
\begin{eqnarray}\left (1-\frac{1}{N+1} \right )1/d \leq c/m<  1/d.\end{eqnarray}
\end{theorem}

\begin{proposition}\label{prop:multiplyadddivisibility}
Consider an integer divisor $d>0$. We have that $d$ divides $n\in[0, N]$  if and only if 
$\remainder(c*n +c, m)  < c$
subject to the condition that
\begin{eqnarray}\left (1-\frac{1}{N+1} \right )1/d \leq c/m<  1/d.\end{eqnarray}\end{proposition}

We can check that the conditions of Theorem~\ref{thm:quotientrobison} are met when $c = \floor(m/d)$ and $m \geq d * (N - \remainder(N,d) +1)$ as long as $m$ is not divisible by $d$. Proposition~\ref{prop:multiplyadddivisibility} is satisfied with the more stringent inequality  $m \geq d * (N+1)$.

\section{Complementarity}
\label{sec:complementarity}

When processing numerators and divisors in $[0,n]$, it may be most convenient if the constant $c$ is also in $[0,n]$.
In this respect, the multiply-shift and multiply-add-shift results are complementary as first shown by Robison~\cite{Robison:2005:NUD:1078021.1078059}.
Suppose that we want to divide all integers $n\in [0,N]$ by $d$ in the case where $N+1$ is a power of two. For example, we may have $N=2^{64}-1$. We want the constant $m$ to be a power of two.

When $d$ is a power of two, efficient division and remainder routines are available. The quotient requires a single binary shift while the remainder requires selecting the low-weight bits with a mask. Thus suppose that the divisor $d$ is not a power of two. 

To satisfy the constraints of Theorems~\ref{thm:quotient} and~\ref{thm:remainder}, we can pick $c = \ceiling (m/d)$ and $m = 2^{\ceiling(\log_2(d))} * (N+1) $.
Unfortunately,  $c = \ceiling (m/d)$ is 
not in $[0, N]$
which may cause implementation
issues. Indeed, if we want to do 64-bit arithmetic on hardware with 64-bit machine words, it is most convenient if all constants fit in 64-bit words.
Thus we may try a smaller constant. The choice $m = 2^{\floor(\log_2(d))} * (N+1) $ is convenient since  $c = \ceiling (m/d)$ is then an integer $\in [0, N]$.
Unfortunately, it is not a valid choice for all divisors $d$, as per the requirements of Theorems~\ref{thm:quotient} and~\ref{thm:remainder}. For example, if $N+1= 2^{32}$ and $d=19$,
picking $m = 2^{\floor(\log_2(d))} * (N+1) =2^{4+32}$, we get $c=  \ceiling (m/d)=3616814566$. We have that $N-\remainder(N+1,d)=2^{32}-6$. We can verify that $3616814566 \geq 2^{4+32}/19 * \left (1 + \frac{1}{ (2^{32}-6)} \right )\approx 3616814565.89$. 
%
Thus the conditions of Theorem~\ref{thm:quotient} are not satisfied.

Thankfully, we can fall back on the multiply-add-shift results. Suppose that setting $c = \ceiling (m/d)$ and $m = 2^{\floor(\log_2(d))} * (N+1) $ fails to satisfy the conditions of Theorem~\ref{thm:quotient}, then we have that 
\begin{eqnarray} \ceiling (m/d) \geq \left (1 + \frac{1}{N- \remainder(N+1,d)} \right )m/d.\label{eq:thm1fails}\end{eqnarray} 
It may be convenient to simplify this equation further. We can multiply both sides by $d$.
We have that $d *  \ceiling (m/d) = m + d - \remainder(m,d) $ on the left-hand-side.  On the right-hand-side, we have $m$ plus some quantity that may not be integer, but we can safely apply the $\ceiling$ function since the left-hand-side is an integer. After subtracting $m$ from both sides, we get
\begin{eqnarray}  d - \remainder(m,d)  \geq   \ceiling\left (\frac{m}{N- \remainder(N+1,d)}\right ).\label{eq:simplerchk} \end{eqnarray} 
(E.g., with $d=19$ and $N+1=2^{32}$, we get
$18\geq 17$.)
We want to show that the conditions of Theorem~\ref{thm:quotientrobison} are satisfied when keeping $m = 2^{\floor(\log_2(d))} * (N+1)$ and setting $c= \floor (m/d)$. 
That is, if Theorem~\ref{thm:quotient} fails us, we can use Theorem~\ref{thm:quotientrobison} so that it is always possible to pick $c\in [0,N]$.

We have that $2^{\floor(\log_2(d))} > d/2$
and hence $m > d * (N+1) / 2$ and therefore
$d/m < 2/(N+1)$.
We assume that
$d$ is not a power of two.
Since we assume $N+1$ and hence $m$ are powers of two,
we have that
$ \ceiling(m/d) = \floor (m/d) +1$ and thus 
from (\ref{eq:thm1fails})  
\begin{eqnarray}  \floor (m/d) +1 \geq \left (1 + \frac{1}{N- \remainder(N+1,d)} \right )m/d.\end{eqnarray} 
We can divide by $m$ and subtract $1/m$ to get
\begin{eqnarray} \floor (m/d)/m & \geq & \left (1 + \frac{1}{N- \remainder(N+1,d)}\right )1/d -1/m \\
& = & \left (1 + \frac{1}{N- \remainder(N+1,d)} - \frac{d}{m} \right )1/d \\
& \geq & \left (1 + \frac{1}{N +1} - \frac{d}{m}  \right )1/d \\
& > & \left (1 + \frac{1}{N+1} - \frac{2}{N+1}  \right )1/d \\
& = & \left (1- \frac{1}{N+1}  \right )1/d\\
& \geq & \left (1- \frac{1}{N- \remainder(N+1,d)+1}  \right )1/d.
\end{eqnarray}

We have shown Proposition~\ref{Proposition:complementary}
, which tells us 
that it is always possible to compute the quotient using $c\in [0,N]$,\footnote{This result is not novel~\cite{Robison:2005:NUD:1078021.1078059}.}
selecting the approach using Equation~\ref{eq:simplerchk}.
\begin{proposition}\label{Proposition:complementary}
Consider an integer divisor $d>0$ that is not a power of two. Let $N$ be an integer such that $N+1$ is a power of two, then we can compute the quotient of any integer $n\in[0, N]$ by $d$  using a constant $c\in [0,N]$ as follows. Let $m = 2^{\floor(\log_2(d))} * (N+1) $. 
\begin{itemize}
    \item if $d - \remainder(m,d) <  \ceiling(m/(N- \remainder(N+1,d)))$, we let $c = \ceiling (m/d)$ and we have $\division(n,d) =\division(c*n , m)$.
    \item Otherwise we  let $c = \floor (m/d)$ and we have $\division(n,d) =\division(c*n +c, m)$.
\end{itemize}
\end{proposition}

The same complementarity exists between the novel theorems for the computation of the quotient and remainder: Theorems~\ref{thm:remainder} and~\ref{thm:remainderrobison}. Indeed, 
 if we choose   $m = 2^{\floor(\log_2(d))} * (N+1) $ and  $c = \ceiling (m/d)$, but the conditions of Theorem~\ref{thm:remainder} are not met, then we have that 
\begin{eqnarray} \ceiling (m/d) \geq \left (1 + \frac{1}{N} \right )m/d.\end{eqnarray} 
We  have that
$ \ceiling(m/d) = \floor (m/d) +1$ since $m$ is not divisible by $d$ and thus
\begin{eqnarray}  
\floor (m/d)& \geq &\left (1 + \frac{1}{N} \right )m/d -1 \\
& = & \left (1 + \frac{1}{N} - \frac{d}{m} \right )1/d \\
& > & \left (1 + \frac{1}{N+1} - \frac{d}{m}  \right )1/d  \label{fiddled}\\
& > & \left (1 + \frac{1}{N+1} -  \frac{2}{N+1}  \right )1/d \\
& = & \left (1- \frac{1}{N+1}  \right )1/d.
\end{eqnarray}

Thus, again, it is always possible to pick $c\in[0,N]$: if not with Theorem~\ref{thm:remainder}, then with Theorem~\ref{thm:remainderrobison},
using the analog of Equation~\ref{eq:simplerchk} to select the
correct approach.
We formalize the result with the novel Proposition~\ref{Proposition:complementarynovel}.

\begin{proposition}\label{Proposition:complementarynovel}
Consider an integer divisor $d>0$ that is not a power of two. Let $N$ be an integer such that $N+1$ is a power of two, then we can compute the quotient and the remainder of any integer $n\in[0, N]$ by $d$  using a constant $c\in [0,N]$ as follows. Let $m = 2^{\floor(\log_2(d))} * (N+1) $. 
\begin{itemize}
    \item if $d - \remainder(m,d) <  \ceiling(m/N)$,\\ we let $c = \ceiling (m/d)$ and we have $\division(n,d) =\division(c*n , m)$ and $\remainder(n,d) =\division(\remainder(c*n, m) * d , m).$
    \item Otherwise we  let $c = \floor (m/d)$ and we have $\division(n,d) =\division(c*n +c , m)$ and $\remainder(n,d) =\division(\remainder(c*n +c, m) * d , m)$.
\end{itemize}
\end{proposition}

\paragraph{Ideal divisors}
When $N+1$ is a power of two, we have shown that it is always possible to pick $m = 2^{\floor(\log_2(d))} * (N+1) $.
However, for some divisors, we can pick an even smaller $m$, even with the constraint that $m$ be a power of two. 
Suppose  that you would like to compute both the remainder and the quotient, as in Theorem~\ref{thm:remainder}. Picking $m= N+1$ would be especially convenient. 
(In architectures where the product is stored in a pair of registers, 
$\division(x,m)$ and $\remainder(x,m)$ are essentially free if $m=2^W$, where
$W$ is the architecture's word size.)
We choose $m = (N+1) $  and $c = \ceiling ((N+1)/d)$. To satisfy the conditions of Theorem~\ref{thm:remainder}, we need that
$c *d < (1+1/N)*(N+1) = N + 2 + 1/N$. We have that
$\ceiling ((N+1)/d) * d = (N+1) + d - \remainder(N+1,d) $. Assuming that $d$ does not divide $N+1$, the inequality holds if and only if $d - \remainder(N+1,d)  = 1$ which is true if and only if $d $ divides $N+2$. We refer to any such divisors as being \emph{ideal},
as they enable us to pick $m=N+1$
. See Table~\ref{tab:ideal}.\footnote{They are related to Fermat numbers~\cite{jaroma2007classical}.}

We can show that we cannot pick $m$ to be a smaller power of two---unless $d$ divides $N+1$. Indeed, if we pick $m=(N+1)/2$, then the conditions of Theorem~\ref{thm:remainder} require that $c *d < (1+1/N)*(N+1)/2$
but $1/N * (N+1)/2<1$ for $N>1$ and since $c*d$ is an integer, we then must have that $c*d\leq (N+1)/2$. Yet we have $\ceiling ((N+1)/(2*d)) * d>  (N+1)/d$ when $d$ does not divide $N+1$.

\begin{table}\caption{Some ideal divisors\label{tab:ideal}}\centering
\begin{tabular}{cl}
\toprule
Range $[0,N+1)$ & ideal divisor \\\midrule
$[0,2^{32})$ &  641 \\
$[0,2^{32})$ &  6700417\\
$[0,2^{64})$ &  274177 \\
$[0,2^{64})$ &  67280421310721  \\\bottomrule
\end{tabular}
\end{table}


\section{Rounding}

Instead of computing the integer division ($\floor(n/d)$), we sometimes wish to round the result to the nearest integer. 
Unsurprisingly and maybe obviously, it is possible to do with an expression of the form $\floor(z/d)$ for some integer $z$ that depends on $n$ and $d$. Hence, our efficient integer quotient computations extend to the computation of the rounded division. Indeed, we have that $\floor((n+ \floor(d/2))/d)$ is the round-to-nearest function; when $d$ is odd and $n$ is between two multiple of $d$, it rounds up.
To round down, we can use  $\floor((n+ \ceiling(d/2) - 1)/d)$ instead.\footnote{We define $\ceiling(x)$ as the smallest integer that is no smaller than $x$.}

It is also possible to handle more complicated scenarios. For example, what if we wish to round to the nearest integer, rounding to the nearest even integer when we are in-between two integers? It is only relevant when the division $d$ is even. Let $z=n+ \floor(d/2)$. 
Whenever $z$ is a multiple of $d$  and $\floor(z/d)$ is odd, we return $\floor(z/d)-1$, else we return $\floor(z/d)$.
We can check that an integer is a multiple of $d$ efficiently with Proposition~\ref{prop:divisibility} for example.
The division and the divisibility test can reuse the same intermediate computations. 




\section{Conclusion}
\label{sec:conclusion}

Our work shows that a unified approach, with the same precomputed constants, allows the  computation of the quotient and remainder, while further providing fast divisibility checks. Thus, for example, an algorithm could check efficiently whether an integer is divisible by another and, in the negative case, compute the remainder while reusing the prior work.

Future work might address the problem of computing generalized expressions such as $\floor(n * x)$ for integer values $n$ and real numbers $x$.
For example, we have that $\floor(\log_2( 5^x ))=  \floor(x * \log_2( 5))$ can be computed as multiplication and a division by a power of two: 
 $(152170 * x) \div 2^{16} $ for $x \in (-400,350)$. We find such optimizations hand-coded in highly optimized algorithms~\cite{number2021}: it might prove useful to formalize the derivation of such routines.

\bibliographystyle{model1-num-names}
\bibliography{wordmodulo}

\end{document}